\documentclass{article}[11pt]

\usepackage{amsmath,amssymb,amsfonts,amsthm}

\iftrue
\usepackage{fullpage}
\newcommand{\reduceitem}{}

\fi

\newtheorem{theorem}{Theorem}[section]    
\newtheorem{fact}[theorem]{Fact}    
\newtheorem{definition}{Definition}[section] 
\newtheorem{claim}[theorem]{Claim}    
\newtheorem{corollary}[theorem]{Corollary}    
\newtheorem{lemma}[theorem]{Lemma}    

\newsavebox{\fmbox}
\newenvironment{fmpage}[1]
     {\medskip\begin{lrbox}{\fmbox}\begin{minipage}{#1}}
     {\end{minipage}\end{lrbox}\fbox{\usebox{\fmbox}}\medskip}
\newcommand{\encadre}[1]{
\begin{center}
\begin{fmpage}{10cm}
#1
\end{fmpage}
\end{center}}

\newcommand{\e}{{\mathrm{e}}}
\newcommand{\complexi}{{\mathrm{i}}}

\newcommand{\E}{\mathbb{E}}
\newcommand{\complex}{\mathbb{C}}
\newcommand{\Span}{\mathsf{Span}}
\newcommand{\ket} [1] {\lvert #1 \rangle}
\newcommand{\bra} [1] {\langle #1 \rvert}
\newcommand{\braket}[2]{\langle #1\vert #2 \rangle}
\newcommand{\density}[1]{\ket{#1}\!\bra{#1}}
\newcommand{\proj}[1]{\ket{#1}\!\bra{#1}}

\newcommand{\norm}[1]{\lVert #1 \rVert}
\newcommand{\size}[1]{\left\lvert #1 \right\rvert}
\newcommand{\abs}{\size}

\newcommand{\tensor}{\otimes}

\newcommand{\diag}{{\mathrm{diag}}}
\newcommand{\Order}{{\mathrm{O}}}
\newcommand{\order}{{\mathrm{o}}}
\newcommand{\adjoint}{\dagger}
\newcommand{\set}[1]{\left\{ #1 \right\}}

\newcommand{\etal}{{\em et al.\/}}

\newcommand{\suppress}[1]{}
\newcommand{\comment}[1]{}

\newcommand{\reflex}{\mathrm{ref}}

\newcommand{\id}{\mathrm{Id}}
\newcommand{\eps}{\varepsilon}

\newcommand{\aitch}{{\mathcal{H}}}
\newcommand{\ay}{{\mathcal{A}}}
\newcommand{\bee}{{\mathcal{B}}}

\newcommand{\kay}{{\mathcal{K}}}

\newcommand{\algof}{\mathbf{Rotate}}
\newcommand{\algod}{\mathbf{Detect}}

\newcommand{\malgod}{\mathbf{MainDetect}}
\newcommand{\tulsi}{\mathbf{T}}
\newcommand{\eht}{\mathsf{HT}}

\newcommand{\qht}{\mathsf{{QHT}}}

\newcommand{\tphi}{{\widetilde{\phi}}}

\begin{document}

\title{On the hitting times of quantum versus random walks
\thanks{
Research supported in part by the European Commission IST Integrated Project
Qubit Applications (QAP) 015848, the 
ANR Blanc AlgoQP grant of the French Research Ministry, 
NSERC Canada, CIFAR,
an Ontario ERA, QuantumWorks, and ARO/NSA (USA). 
}}
\author{Fr\'ed\'eric Magniez\thanks{%
LRI, Univ. Paris-Sud, CNRS;
F-91405 Orsay, France;
\texttt{magniez@lri.fr},
\texttt{richterp@lri.fr}.
 }
\and
Ashwin Nayak\thanks{%
C\&O and IQC, U.\ Waterloo, and Perimeter Institute,
Ontario N2L 3G1, Canada;
\texttt{anayak@uwaterloo.ca}.
Some of this work was conducted while this author was visiting LRI, U.\
Paris-Sud, Orsay, France.
Research at Perimeter Institute is supported in part by the
Government of Canada through NSERC and by the Province of Ontario
through MEDT. 
}
\and 
Peter C. Richter$^{\dag}$
\and 
Miklos Santha\thanks{LRI, Univ. Paris-Sud, CNRS;
F-91405 Orsay, France and
Centre for Quantum Technologies,
National University of Singapore, Singapore 117543; \texttt{santha@lri.fr}.
Research at the Centre
for Quantum Technologies is funded by the Singapore Ministry of Education 
and the National Research Foundation.}
}

\date{}
\maketitle

\begin{abstract}
  The {\em hitting time} of a classical random walk (Markov chain) is
  the time required to {\em detect} the presence of -- or
  equivalently, to {\em find} -- a marked state.  The hitting time of
  a quantum walk is subtler to define; in particular, it is unknown
  whether the detection and finding problems
have the same time complexity.  In this
  paper we define new Monte Carlo type classical and quantum hitting
  times, and we prove several relationships among these and the
  already existing Las Vegas type definitions.  In particular, we show
  that for some marked state the two types of hitting time 
  are of the same order in both the classical and the quantum case.

  Further, we prove that for any reversible ergodic Markov chain $P$,
  the quantum hitting time of the quantum analogue of $P$ has the same
  order as the square root of the classical hitting time of $P$.  We
  also investigate the (im)possibility of achieving a gap greater than
  quadratic using an alternative quantum walk.
  In doing so, we define a notion of reversibility for
a broad class of quantum
  walks and show how to derive from any such quantum walk a classical
  analogue.
  For the special case of quantum walks built on reflections, we show
  that the hitting time of the classical analogue is exactly the
  square of the quantum walk.

  Finally, we present new quantum algorithms for the detection and
  finding problems.  The complexities of both algorithms are related
  to the new, potentially smaller, quantum hitting times.  The
  detection algorithm is based on phase estimation and is particularly
  simple.  The finding algorithm combines a similar phase estimation
  based procedure with ideas of Tulsi from his recent
  theorem~\cite{Tulsi08} for the 2D grid.  Extending his result, we
  show that for any state-transitive Markov chain with unique marked
  state, the quantum hitting time is of the same order for both the
  detection and finding
  problems.  
\end{abstract}


\section{Introduction}

Many classical randomized algorithms are based on {\em random walks},
or {\em Markov chains}.  Some use random walks to generate random
samples from the Markov chain's stationary distribution, in which case
the {\em mixing time} of the Markov chain is the complexity measure of
interest.  Others use random walks to search for a ``marked'' state in the
Markov chain, in which case the {\em hitting time} is of interest.  In
recent years, researchers studying {\em quantum walks} have attempted
to define natural notions of ``quantum mixing time''
\cite{NV,ABNVW,AAKV} and ``quantum hitting time'' \cite{AKR,Sze,MNRS}
and to develop quantum algorithmic applications to sampling and search
problems.

Ambainis \cite{Amb1} was the first to solve a natural problem---the
``element distinctness problem''---using a quantum walk.  Following
this, quantum walk algorithms were discovered for triangle finding
\cite{MSS}, matrix product verification \cite{BS}, and group
commutativity testing \cite{MN}.  All of these are ``hitting time''
applications involving quantum walk search on Johnson
graphs---highly-connected graphs whose vertices are subsets of a fixed
set and whose edges connect subsets differing by at most two elements.
Quantum walk algorithms for the generic {\em spatial search} problem
\cite{AA} were given by Shenvi \etal~\cite{SKW} on the hypercube, and
by Childs and Goldstone \cite{CG2} and Ambainis \etal~\cite{AKR} on
the torus.  Szegedy \cite{Sze} showed that for any symmetric Markov
chain and any subset $M$ of marked elements, we can detect whether or
not $M$ is nonempty in at most 
(of the order of) the square-root of the classical hitting time.  To
achieve this goal, Szegedy defined the quantum analogue of any
symmetric Markov chain.  Later Magniez \etal~\cite{MNRS} extended this
to define the quantum analogue of the larger class of
irreducible Markov chains.

Unresolved by Szegedy's work is the question: with what probability
does the algorithm output a marked state, as opposed to merely
detecting that $M$ is nonempty?  This issue was addressed by Magniez
\etal~\cite{MNRS}, who gave an algorithm which finds a marked state
with constant probability but whose complexity may be more than the
square root of the classical hitting time. Indeed, for the $\sqrt{N}
\times \sqrt{N}$ grid 
their bound is $\Theta(N)$ whereas the classical hitting time is
$\Theta(N \log N)$.  The algorithms of Ambainis \etal~\cite{AKR} and
Szegedy~\cite{Sze} perform better on the grid if there is a {\em
  unique\/} marked state: they find the marked state in time
$\Order(\sqrt{N} \cdot \log N)$.  
(The case of multiple marked elements may be reduced to this case at
the cost of a polylog factor in run-time.)  For some time it remained
unclear if one could do better, until Tulsi \cite{Tulsi08} showed how
to find a unique marked element in time $\Order(\sqrt{N \log N})$.  His
algorithm seems to be something of a departure from previous quantum
walk algorithms, most of which have been analyzable as the product of
two reflections 
{\em {\`a} la\/} the Grover algorithm \cite{Gro}. 
The 2D grid
was the canonical example of a graph on which it was unknown how to
find a marked state quantumly with the same complexity as detection.
Tulsi's result thus raises the question: is finding ever any harder than
detection?


In this paper we address several questions related to classical and
quantum hitting times.  In the literature on Markov chains,
hitting time is usually defined as the complexity of the natural Las
Vegas algorithm for finding a marked element by running the chain.  We
first give an alternative definition based on the Monte Carlo version
of the same algorithm. To our knowledge, this variant of the hitting
time has not been considered previously.  We show that for some marked
state, the two hitting times are of the same order
(Theorem~\ref{classicalcomparison}).

Within the setting of abstract search algorithms presented by Ambainis
\etal~\cite{AKR}, we introduce quantum analogues of the two classical
hitting times (Definition~\ref{qhtdef}).  The analogue of the Las
Vegas version was already present in
Szegedy's work~\cite{Sze},
whereas the other is new.  Unlike in the classical case, detection and
finding are substantially different problems in quantum computing.  We
address both problems here.

For the detection problem, we introduce a new algorithm $\algod$ based
on phase estimation which is similar to the approach of Magniez
\etal~\cite{MNRS}.  
Our algorithm
can detect the presence of a marked element in the starting state.
The advantages of this algorithm are its
simplicity and the fact that its complexity is related to the new
Monte Carlo type quantum hitting time (Theorem~\ref{detecting}).
\suppress{
Coupled with a standard control test algorithm, our algorithm can
detect the presence of a marked element in the starting state
(Theorem~\ref{detecting}).  }
This is an improvement over the
Szegedy detection algorithm whose complexity
is related to the potentially larger Las Vegas quantum hitting time.

We then present a variant of the above algorithm, called $\algof$,
which can be used for the more difficult problem of finding, and whose
complexity is also related to the Monte Carlo type quantum hitting
time (Theorem~\ref{finding}).  This improves the finding algorithm
due to Ambainis \etal\ whose complexity was characterized by a potentially
larger quantity, the inverse of the smallest eigenphase of the search
algorithm.  Our algorithm is also simpler.  Combining $\algof$ with
the ideas
in the Tulsi algorithm for the 2D grid, we can find a marked
element with constant probability and with the same complexity as
detection for a large class of quantum walks---the quantum analogue of
state-transitive reversible ergodic Markov chains.

As in the classical case, for some marked elements the two
types of the quantum hitting time are of the same order
(Fact~\ref{factqcomparison} and Theorem~\ref{qcomparison}).  For any
reversible ergodic Markov chain $P$, we prove that the quantum hitting
time of the quantum analogue of $P$ is of same order as the square
root of the classical hitting time of $P$
(Theorem~\ref{thm-qht-sqrt}).  Moreover, for the Las Vegas hitting
times they are exactly the same.

Finally, we investigate the (im)possibility of achieving a greater than quadratic gap using some other quantum walk.
For this we consider general quantum walks on the edges of an
undirected graph $G$; these were defined, for example, in the survey
paper of Ambainis~\cite{amb}, see also~\cite{sant}.
We define a quite natural notion of
reversibility for general quantum walks.
We conjecture that for any reversible quantum walk $U_2$ on an undirected graph $G$, 
there exists a reversible ergodic Markov chain $P$ on $G$
such that for every marked state, the quantum hitting time of $U_2$
is at most the square root of the 
classical hitting time of $P$.  We are able to prove this in the
special case of quantum walks built on reflections
(Theorem~\ref{lowerbound}).  Our proof introduces a classical analogue
of
such quantum walks which might be of independent interest
(Definition~\ref{classicalization}).
Curiously, the classical analogue
is reversible if and only if the quantum walk is reversible
(Lemma~\ref{classical-analogue}).

\section{Classical hitting times}\label{classical}

Let $P$ be an ergodic and reversible Markov chain over state space $X
= \{1, \ldots , n\}$, 
which we identify with its transition probability matrix. We suppose
that the eigenvalues of $P$ are nonnegative, by replacing $P$ with
$(P+I)/2$ if necessary.  More generally, one may also assume that the
eigenvalues of $P$ are at
least~$\alpha$, where~$\alpha \in [0,1)$, by replacing $P$ with
$((1-\alpha) P+(1+\alpha)I)/2$ if necessary.  We make this further
assumption when needed, for instance in Section~\ref{find}.  Let $\pi$
denote the stationary distribution of $P$.  Let $P_{-z}$ be the $(n -
1) \times (n-1)$ matrix we get by deleting from $P$ the row and column
indexed by $z$. Similarly, for a vector $v$, we let $v_{-z}$ stand for
the vector obtained by omitting the $z$-coordinate of $v$.

\begin{claim}
The eigenvalues of $P_{-z}$ are all in the
interval $[\kappa_n,1)$, where $\kappa_n$ is the smallest eigenvalue of $P$.
\end{claim}
\begin{proof}
  The proof globally proceeds along the lines of the proof of Lemma~8
  in~\cite{Sze}.
  For $x \in X$ let $e_x$ denote the characteristic vector of $x$.
  Let $w_1, \ldots , w_n$ be the eigenvectors of $P$ with associated
  eigenvalues $1= \kappa_1 \geq \ldots \geq \kappa_n > 0$. Let $v$ be
  an arbitrary eigenvector of $P_z$ with eigenvalue $\lambda$.
  Since~$P$ is ergodic, $\norm{P_z} < 1$, therefore~$\lambda < 1$.  We
  show that $\lambda \geq \kappa_n$. This is obviously true if
  $\lambda = \kappa_k$ for some $k$; let us suppose that this is not
  the case.

Let $w$ be the vector obtained from $v$ by augmenting it with a 0
in the~$z$-coordinate.  We express both $w$ and $e_z$ in the
eigenbasis of $P$: let $w = \sum_{k=1}^n \gamma_k w_k$ and $e_z =
\sum_{k=1}^n \delta_k w_k$. Then $wP = \lambda w + \nu e_z$ for some
real number $\nu$. Moreover, $\nu \neq 0$; otherwise $w$ would
have been an eigenvector of $P$, meaning that $\lambda=\kappa_k$,
which contradicts our supposition.  
For $k = 1, \ldots , k$, we have $\kappa_k \gamma_k = \lambda \gamma_k
+ \nu \delta_k$.  Since $w$ and $e_z$ are orthogonal, we also have
$\sum_{k=1}^n \gamma_k \bar{\delta_k} = 0$. Therefore $\sum_{k = 1}^n
\frac{|\delta_k|^2}{\kappa_k - \lambda} = 0$.  The statement then
follows since the left hand side of the above equation would be
positive if $\lambda$ were less than $\kappa_n$.
\end{proof}

\begin{definition}
For $z \in X$, the {\em $z$-hitting time} of $P$, denoted
by $\eht(P,z)$, is the expected number of steps the chain $P$ takes to reach the state $z$ when started
in the initial distribution $\pi$. 
\end{definition}
It is well known that the $z$-hitting time of $P$ is given by the formula
$\eht(P,z) = \pi_{-z}^\dagger (I-P_{-z})^{-1} u_{-z}$,
where $u$ is the all-ones vector. 
Simple algebra
shows that
\[
\pi_{-z}^\dagger (I-P_{-z})^{-1} u_{-z} 
    \quad = \quad \sqrt{\pi_{-z}}^\dagger (I-S_{-z})^{-1} \sqrt{\pi_{-z}},
\]
where $\sqrt{\pi_{-z}}$ is the entry-wise square root of $\pi_{-z}$
and $S_{-z}$ is the ``symmetrized form'' $S_{-z} = \sqrt{\Pi_{-z}}
P_{-z} \sqrt{\Pi_{-z}}^{-1}$ of $P_{-z}$ with $\Pi_{-z} =
\diag(\pi_{x})_{x\neq z}$.  The matrices $P_{-z}$ and $S_{-z}$ have
the same spectrum since they are similar.  Let $\{v_j : j \leq n-1 \}$
be the set of normalized eigenvectors of $S_{-z}$ where the eigenvalue
of $v_j$ is $\lambda_j = \cos \theta_j$ with $0 < \theta_j \leq
\pi/2$.  By reordering the eigenvalues we can suppose that $1 >
\lambda_1 \geq \ldots \geq \lambda_{n-1} \geq 0$.  If $\sqrt{\pi_{-z}}
= \sum_j \nu_j v_j$ is the decomposition of $\sqrt{\pi_{-z}}$ in the
eigenbasis of $S_{-z}$ then the $z$-hitting time satisfies:
\[
\eht(P,z)  \quad = \quad \sum_j \frac{\nu_j^2}{1-\lambda_j}.
\]
When $0 < \theta \leq \pi/2$ then $1 - \theta^2/2 \leq \cos \theta
\leq 1 - \theta^2/4.$ Therefore we can approximate the hitting time
with another expectation that is very closely related to the analogous
quantum notion. More precisely, let $H_z$ be the random variable which
takes the value $1/\theta_j^2$ with probability $\nu_j^2$, and 0 with
probability $1 - \sum_j \nu_j^2$. We denote the expectation of $H_z$
by $\E[H_z]$. Then we have $ 2\, \E[H_z] \leq \eht(P,z) \leq 4\, \E[H_z]$.

In the definition of the hitting time the Markov chain is used 
in a Las Vegas algorithm: we count the (expected) number of steps to
reach the marked element without error.  We can also use the chain as
an algorithm that reaches the marked element with some probability
smaller than~$1$, leading to a Monte Carlo type definition.
Technically, to be able to underline the analogies between the
classical and quantum notions, we define the hitting time with error
via the random variable $H_z$.
\begin{definition}
For $ z \in X$  and for $0 < \eps <1$, the {\em $\eps$-error 
$z$-hitting time} of $P$, denoted by $\eht_{\eps}(P,z)$ is defined as
\[
\eht_{\eps}(P,z) \quad = \quad \min \;\; \{ y \;:\; \Pr [ H_z > y ] \leq \eps \}.
\]
\end{definition}

Observe that for all $z$, if $\eps \leq \eps'$ then $\eht_{\eps'}(P,z)
\leq \eht_{\eps}(P,z)$. We first show that the use of $H_z$ for the
definition of the Monte Carlo hitting time is indeed justified (proof
in Appendix~\ref{appjustif}).  For this, let us denote by
$h_{\eps}(P,z)$ the smallest integer $k$ such that the probability
that the chain does not reach $z$ in the first $k$ steps is at most
$\eps$.

\begin{theorem}\label{justif}
For all $z$ and $\eps$, we have
\begin{eqnarray*}
h_{\eps}(P,z) & \leq & \left(4 \ln \frac{2}{\eps} \right) 
                         \eht_{\eps/2}(P,z), \quad \text{and} \\
\eht_{\eps}(P,z) & \leq &  \frac{1}{2} \; h_{\eps/3}(P,z).
\end{eqnarray*}
\end{theorem}

How much smaller than the Las Vegas hitting time can the Monte Carlo
hitting time be?  The following results state that for some $z$ they
are of the same order of magnitude.

\begin{theorem}\label{classicalcomparison}
We have the following inequalities between the two notions of hitting time:
\begin{itemize}\reduceitem
\item For all $z$ and $\eps$, \ \ $\eht_{\eps}(P,z) \leq \frac{1}{2 \eps}\; \eht(P,z)$.
\item There exists $z$  such that for all $\eps < 1/2$,  \ \ 
$\eht(P,z) \leq 4\, \eht_{\eps}(P,z)$. 
\end{itemize}
\end{theorem}

\begin{proof}
  The first statement simply follows from the Markov inequality and from
  the relation $ \E[H_z] \leq \eht(P,z)/2$.  For the second statement,
  let $z$ be an element such that $\nu_1^2 \geq 1/2$. The existence of
  such an element is assured by Lemma~8 in~\cite{Sze}. Then
$\eht(P,z) \leq \sum_j 4 \nu_j^2 / \theta_j^2 \leq 4/\theta_1^2 \leq 4 \eht_{\eps}(P,z)$.
\end{proof}

\section{Quantum hitting times}
\subsection{Two notions of quantum walk}\label{ct}

Let $U=U_2 U_1$ be an {\em abstract search algorithm} as
in~\cite{AKR}, where $U_1 = I - 2\proj{\mu}$ for a ``target vector''
$\ket{\mu}$ with real entries, and $U_2$ 
is a real unitary matrix with a unique
$1$-eigenvector $\ket{\phi_0}$.  Without loss of generality we always
assume that $\ket{\phi_0}$ has real entries.  The state $\ket{\mu}$ is
the quantum analogue of the state $z$
which we seek in the
classical walk $P$, $U_2$ the analogue of $P$, and $\ket{\phi_0}$ the
analogue of the stationary distribution $\pi$.

The abstract search algorithm usually starts with 
state~$\ket{\phi_0}$, and iterates $U$ several times in order to get a
large deviation from $\ket{\phi_0}$.  In this paper, we prefer to
start with a slightly different initial state.  The general behavior
of the abstract search algorithm remains unchanged by this.  We
replace the starting state $\ket{\phi_0}$ by $\ket{\tphi_0} =
\ket{\phi_0} - \braket{\phi_0}{\mu}\ket{\mu}$, the (unnormalized)
projection of the~$1$-eigenvector~$\ket{\phi_0}$ of~$U_2$ on the space
orthogonal to~$\ket{\mu}$. This substitution was first considered
in~\cite{Sze}, and corresponds to first making a measurement according
$(\ket{\mu},\ket{\mu}^\perp)$. If the measurement outputs $\ket{\mu}$
we are done.  Otherwise we run the abstract search algorithm.

This choice of the initial state is motivated by the results in
Section~\ref{sec-quadratic} which relate quantum hitting time to
classical hitting time. All other results in 
this paper remain valid if we keep $\ket{\phi_0}$.

Ambainis \etal\ characterized the spectrum of $U$ in term of the
decomposition of $\ket{\mu}$ in the eigenvector basis of $U_2$. One of
their results is:
\begin{theorem}[\cite{AKR}]\label{akr0}
Let $U_2$ be a unitary
matrix with real entries and a unique $1$-eigenvector $\ket{\phi_0}$.
Let $\ket{\mu}$ be a unit vector with real entries, and let $U_1= I - 2\proj{\mu}$.
Let $U=U_2U_1$.
\begin{itemize}\reduceitem
\item If $\braket{\phi_0}{\mu}=0$, then $\ket{\tphi_0}=\ket{\phi_0}$ and $U\ket{\tphi_0}=\ket{\tphi_0}$.
\item If $\braket{\phi_0}{\mu}\neq 0$,
then
$U$ has no $1$-eigenspace.
\end{itemize}
\end{theorem}
Thus one can use $U$ in order to detect if $\braket{\phi_0}{\mu}\neq
0$.  Indeed, in that case, after a certain number $T$ of iterations of
$U$ on $\ket{\tphi_0}$, the state moves far from the intial state
$\ket{\tphi_0}$.  Such a deviation caused by some operator $V$ (in our
case $V=U^T$,
i.e., $U$ iterated~$T$ times) is usually detected by using the well known {\em control
  test\/} which requires the use the controlled operator $\text{c-}V$.
Namely observe that $(H\otimes I)(\text{c-}V)(H\otimes I)
\ket{0}\ket{\psi}= \frac{1}{2}\ket{0}(\ket{\psi}+V\ket{\psi}) +
\frac{1}{2}\ket{1}(\ket{\psi}-V\ket{\psi})$.  Therefore 
a measurement of the first register gives outcome $1$ with probability
$\norm{\ket{\psi}-V\ket{\psi}}^2/4$.

Szegedy~\cite{Sze} designed a generic method for constructing an
abstract search algorithm given a (classical) Markov chain.  Let
$P=(p_{xy})$ be an ergodic
Markov chain over state space $X = \{1, \ldots , n\}$ with stationary
distribution $\ket{\pi}$.  The time-reversal~$P^*$ of this chain is
defined by equations $\pi_y p^*_{yx}=\pi_x p_{xy}$.
The chain $P$ is reversible if $P=P^*$.

The quantum analogue of $P$ may be thought of as a walk on the {\em edges\/} of the
original Markov chain, rather than on its vertices. Thus, its state space
is a vector subspace of~$\aitch = \complex^{X \times X}$.
For a state~$\ket{\psi} \in \aitch$,
let~$\Pi_\psi = \proj{\psi}$ denote the orthogonal projector
onto~$\Span(\ket{\psi})$, and let $\reflex(\psi)= 2 \Pi_\psi - \id $
denote the reflection through the line generated by $\ket{\psi}$,
where~$\id$ is the identity operator on~${\aitch}$.
If $\kay$ is a subspace of $\aitch$ spanned by a set of mutually orthogonal
states $\{\ket{\psi_i} : i \in I\}$, then let
$\Pi_{\kay} = \sum_{i \in I} \Pi_{\psi_i} $ be the orthogonal projector
onto $\kay$, and let 
$\reflex(\kay) = 2 \Pi_{\kay} - \id $
be the reflection through $\kay$.
Let $\ay=\Span( \ket{x}\ket{p_x} : x\in X)$ and 
$\bee=\Span( \ket{p^*_y}\ket{y} : y\in X)$ be vector subspaces
of~$\aitch$, 
where
\[
\ket{p_x} \quad = \quad \sum_{y\in X} \sqrt{p_{xy}} \, \ket{y}
\quad\textrm{and} \quad 
\ket{p^*_y} \quad = \quad \sum_{x\in X} \sqrt{p^*_{yx}} \,
\ket{x}.
\]
Define similarly for any $z\in X$ the subspaces $\ay_{-z}=\Span(
\ket{x}\ket{p_x} : x\in X\setminus\{z\})$ and $\bee_{-z}=\Span(
\ket{p^*_y}\ket{y} : y\in X\setminus\{z\})$.

\begin{definition}[\cite{Sze,MNRS}]
\label{def-qmc}
  Let $P$ be an ergodic Markov chain.  The unitary operation $W(P) =
  \reflex(\bee) \cdot \reflex(\ay)$ defined on $\aitch$ is called the
  {\em quantum analogue} of~$P$; and the unitary operation $W(P,z) =
  \reflex(\bee_{-z}) \cdot \reflex(\ay_{-z})$ defined on $\aitch$ is
  called the {\em quantum analogue} of~$P_{-z}$.
\end{definition}

The unitary operation $\mathsf{SWAP}$ is defined by
$\mathsf{SWAP}\ket{x}\ket{y}=\ket{y}\ket{x}$. 
When $P$ is
reversible, the connection between the quantum walk of Szegedy and the
quantum walk of Ambainis \etal\ is made explicit by the following
fact.
\begin{fact}
  \label{fact-akr-sze}
Let $z\in X$ and $\ket{\mu}=\ket{z}\ket{p_z}$.  Let $U_2=\mathsf{SWAP}
\cdot \reflex(\ay)$ and $U_1=I-2\proj{\mu}$.  If $P$ is reversible
then $(U_2 U_1)^2=W(P,z)$. In particular, the unitary operators~$U =
U_2 U_1$ and~$W(P,z)$ are diagonal in the same orthonormal basis.
\end{fact}

\subsection{Phase estimation and quantum hitting time}
\label{sec-algorithms}

Let $U$ be a unitary matrix with real entries.  The potential
eigenvalues of $U$ are then $1$, $ -1$, and pairs of conjugate complex
numbers $(\e^{\complexi\alpha_j},\e^{-\complexi\alpha_j})$ 
with $0 < \alpha_j < \pi$, for $1 \leq j \leq J$, for some~$J$.

Let $\ket{\psi}$ be a vector with real entries and of norm at most one.
Then $\ket{\psi}$ uniquely decomposes as
\begin{eqnarray}
\label{eqn-decomp}
\ket{\psi} & = & \delta_0\ket{w_0}+\sum_{1 \leq j \leq J} 
\delta_j (\ket{w_j^+} +  \ket{w_{j}^-})+\delta_{-1}\ket{w_{-1}},
\end{eqnarray}
where $\delta_0,\delta_{-1},\delta_j$ are reals, $\ket{w_0}$ is a unit
eigenvector of~$U$ with eigenvalue $1$, $\ket{w_{-1}}$ is a unit
eigenvector with eigenvalue $-1$, and $\ket{w_j^+}, \ket{w_{j}^-}$ are
unit eigenvectors with respective eigenvalues $\e^{\complexi\alpha_j}$
and $\e^{-\complexi\alpha_j}$, and
$\ket{w_{j}^-}=\overline{\ket{w_j^+}}$.

\suppress{
\begin{definition}
The {\em $U$-flip} of $\ket{\psi}$ is defined as 
$2\delta_0\ket{w_0}-\ket{\psi} 
$.
\end{definition}

\begin{fact}
Let $U=U_2U_1$, where $U_1=I-2\proj{\mu}$ is an abstract search algorithm.
If $\braket{\phi_0}{\mu}=0$, then the $U$-flip of $\ket{\tphi_0}$ is $\ket{\tphi_0}$.
If $\braket{\phi_0}{\mu}\neq 0$, then the $U$-flip of $\ket{\tphi_0}$ is $-\ket{\tphi_0}$.
\end{fact}
}

We now describe a procedure whose purpose is to detect the
state~$\ket{\psi}$ has a component orthogonal to the~$1$-eigenspace
of~$U$. In the context of the abstract search algorithm, this is
equivalent to~$\braket{\phi_0}{\mu}\neq 0$.
The idea, similar to the approach of~Magniez \etal~\cite{MNRS}, 
is to apply the phase estimation algorithm of 
Kitaev~\cite{Kitaev96} and Cleve \etal~\cite{CleveEMM98} to
$U$.  

\begin{theorem}[\cite{Kitaev96,CleveEMM98}]
\label{thm-diffusion1}
Given an eigenvector~$\ket{v}$ of a unitary operator~$U$ with
eigenvalue $\e^{\complexi\alpha}$, the corresponding phase
$\alpha \in (-\pi,\pi]$ can be determined with precision $\Delta$ and
error probability at most $1/3$ by
a circuit $\mathbf{Estimate}$.
If $\ket{v}$ is a $1$-eigenvector of $U$,
then  $\mathbf{Estimate}$ determines $\alpha=0$ with probability $1$.
Moreover, $\mathbf{Estimate}$ makes $\Order(1/\Delta)$ calls to the
controlled operator $\text{c-}U$ and its inverse, and it contains
$\Order((\log{1}/{\Delta})^2)$ additional gates.
\end{theorem}
Based on the circuit \textbf{Estimate}, we can detect the presence of
components orthogonal to the~$1$-eigenspace in
an arbitrary state $\ket{\psi}$.
\encadre{
$\algod(U,\Delta,\eps)$  \quad --- \quad
Input: $\ket{\psi}$
\begin{enumerate}
\item\label{s1}
Apply $ \Theta(\log(1/\eps))$ times
the phase estimation circuit $\mathbf{Estimate}$ for $U$ with precision $\Delta$
to the same state $\ket{\psi}$.
\item If at least one of the estimated phases is nonzero, 
\texttt{ACCEPT}.\\
Otherwise \texttt{REJECT}.
\end{enumerate}
}

Let ${\mathit{QH}}$ be the random variable which takes the value $1/\alpha_j$ with probability
$2\delta_j^2$,  the value $1/\pi$ with probability $\delta_{-1}^2$, 
and
the value~$0$ otherwise. 
Observe that in the following lemma, and in the analysis of all our algorithms, 
the probabilities in fact sum to $\norm{\psi}^2$,
since $\ket{\psi}$ is not necessarily normalized, and has norm at most~$1$.
\begin{lemma}\label{lemma-hitting}
Assume that
$\Pr [ \mathit{QH} > 1/\Delta ] \leq \eps$.
Then the procedure
$\algod(U,\Delta,\eps)$ accepts $\ket{\psi}$ with probability $\norm{\psi}^2-\delta_0^2-\Order(\eps)$,
and moreover with probability $0$ if $\abs{\delta_0}=\norm{\psi}$.
\suppress{
 $\ket{\psi}$ to a state 
at Euclidean distance
$\Order(\sqrt{\eps}\;)$ 
from the $U$-flip of $\ket{\psi}$.}
In addition, the number of applications of $\text{c-}U$ 
is
$ \Order(\log(1/\eps)/\Delta)$.
\end{lemma}
\begin{proof}
Let us first assume that $\mathbf{Estimate}$ can compute the eigenphase of any eigenvector
with certainty. This assumption is in fact valid when $\abs{\delta_0}=\norm{\psi}$.
Then the procedure $\algod$ rejects exactly with probability $\delta_0^2$.

Assume now that $\Pr [ \mathit{QH} > 1/\Delta ] \leq \eps$, and $\abs{\delta_0}<\norm{\psi}$.
First observe that $\mathbf{Estimate}$ with precision $\Delta$ uses $1/\Delta$ applications of $\text{c-}U$.
Then the precision $\Delta$ in $\mathbf{Estimate}$
ensures a nonzero approximation of an eigenphase $\pm\alpha_j$
with probability at least~$2/3$ provided that $\alpha_j\geq \Delta$.
By hypothesis, the contribution of these eigenphases has squared Euclidean norm
$2\sum_j\delta_j^2$.
The success probability is then amplified to $1-\Order(\eps)$ by checking
that all the $\Order(\log (1/\eps))$ outcomes of $\mathbf{Estimate}$ are nonzero.
For the special case of eigenphase $0$, whose contribution has squared
Euclidean norm $\delta_0^2$,
$\mathbf{Estimate}$ gives
approximation~$0$ with probability~$1$.

The contribution of the other eigenphases has squared Euclidean norm
less than $\eps$ in the vector $\ket{\psi}$. 
Therefore the overall acceptance probability is
$\norm{\psi}^2-\delta_0^2 - \Order({\eps})$.
\end{proof}

In the case of quantum walk, the above theorem justifies the following definitions
of quantum hitting times.
Let $U$ be some abstract search $U_2 U_1$, where 
$U_1=I-2\proj{\mu}$,
starting from state $\ket{\tphi_0} = \ket{\phi_0} - a_0 \ket{\mu}$,
where~$a_0 = \braket{\mu}{\phi_0}$. We now set $\ket{\psi}=\ket{\tphi_0}$.
Again,
$\mathit{\mathit{QH}}$ is the random variable which takes the value $1/\alpha_j$ with probability $2\delta_j^2$,  the value $1/\pi$ with probability $\delta_{-1}^2$, 
and $0$ otherwise. 

\begin{definition}\label{qhtdef}
The {\em quantum $\ket{\mu}$-hitting time} of $U_2$ is the expectation of $\mathit{QH}$, that is
$$ \qht(U_2,\ket{\mu}) \quad = \quad  2\sum_j \frac{\delta_j^2 }{\alpha_j}+ \frac{\delta_{-1}^2}{\pi}.$$
For $0 < \eps\ <1$, the {\em quantum $\eps$-error 
$\ket{\mu}$-hitting time} of $U_2$ is defined as
$$\qht_{\eps}(U_2,\ket{\mu}) \quad  = \quad  
\min \{y \;:\; \Pr [ \mathit{QH} > y ] \leq \eps \}.$$
\end{definition}

Using Theorem~\ref{akr0}, Lemma~\ref{lemma-hitting} and our definition of quantum hitting time,
 we directly get:
\begin{theorem}\label{detecting}
For every $T\geq\max\set{1,\qht_\eps(U_2,\ket{\mu})}$, the procedure
$\algod(U,1/T,\eps)$ accepts $\ket{\tphi_0}$ with probability $\norm{\tphi_0}^2-\Order(\eps)$
if $\braket{\phi_0}{\mu}\neq 0$, and
accepts with probability $0$ otherwise.
\suppress{maps $\ket{\tphi_0}$ to a state 
at Euclidean distance
$\Order(\sqrt{\eps}\;)$ 
from the $U$-flip of $\ket{\tphi_0}$.}
Moreover the number of applications of $\text{c-}U$ 
is
$ \Order(\log(1/\eps)\times T)$.
\end{theorem}

If one would like to deal only with normalized states,
and to come back to the original starting state $\ket{\phi_0}$,
we can encapsulate the projection to the space
orthogonal to~$\ket{\mu}$  into our algorithm such as in the following main procedure,
and deduce its behavior from the above theorem.

\encadre{
$\malgod(U_2,\ket{\mu},\Delta,\eps)$  \quad --- \quad
Input: $\ket{\psi}$
\begin{enumerate}
\item Make a measurement according
$(\ket{\mu},\ket{\mu}^\perp)$.
\item If the measurement outputs $\ket{\mu}$, \texttt{ACCEPT}.\\
Otherwise apply $\algod(U,\Delta,\eps)$.
\end{enumerate}
}
\begin{corollary}
For every $T\geq\max\set{1,\qht_\eps(U_2,\ket{\mu})}$, the procedure
$\malgod(U_2,\ket{\mu},1/T,\eps)$ accepts $\ket{\phi_0}$ with probability $1-\Order(\eps)$
if $\braket{\phi_0}{\mu}\neq 0$, and
accepts with probability $0$ otherwise.
\suppress{maps $\ket{\tphi_0}$ to a state 
at Euclidean distance
$\Order(\sqrt{\eps}\;)$ 
from the $U$-flip of $\ket{\tphi_0}$.}
\end{corollary}

\suppress{Combining this theorem together with the control test in Section~\ref{ct}
and Theorem~\ref{akr0}, one can get the following simple algorithm
for detecting if $\braket{\phi_0}{\mu}\neq 0$.
\begin{corollary}\label{algodetect}
For every $T\geq\max(1,\qht_\eps(U_2,\ket{\mu})$, 
the control test applied on $\ket{\tphi_0}$ with the operator
$\algod(U,1/T,\eps)$
outputs
$0$ with probability $1-\Order(\eps)$, if $\braket{\phi_0}{\mu}= 0$;
$1$ with probability $1-\Order(\eps)$,  if $\braket{\phi_0}{\mu}\neq 0$.
\end{corollary}
}

When the abstract search is built on the quantum analogue of a
reversible Markov chain $P$ and $\ket{\mu}=\ket{z}\ket{p_z}$ for some
$z$, we use the following terminology:
\begin{itemize}\reduceitem
\item The {\em quantum $z$-hitting time of $P$} is 
$\qht (P,z) =  \qht(\mathsf{SWAP}\cdot\reflex(\ay), \ket{z}\ket{p_z})$;
\item For $0 < \eps\ <1$, the {\em quantum  $\eps$-error 
$z$-hitting time of $P$} 
is 
$\qht_{\eps}(P,z) = \qht_{\eps}(\mathsf{SWAP}\cdot\reflex(\ay), \ket{z}\ket{p_z})$.
\end{itemize}

With different, more technical arguments, Szegedy proved 
results similar to Theorem~\ref{detecting} albeit with the parameter
$\qht(P,z)$ for symmetric Markov chains:
\begin{theorem}[\cite{Sze}]
When $t$ is chosen uniformly at random in $\{1,2,\ldots,\lceil\qht(P,z)\rceil\}$,
then the expectation of the deviation $\norm{(W(P,z))^t\ket{\tphi_0}-\ket{\tphi_0}}$
is $\Omega(\norm{\tphi_0})$.
\end{theorem}
Under certain
assumptions, Ambainis \etal~\cite{AKR} have a similar result in terms
of the smallest eigenphase of $U_2U_1$.

Suppose we wish to not only detect if~$\braket{\mu}{\phi_0} \neq 0$,
but also to map $\ket{\psi}$ to $\ket{\mu}$.  
Then we are led to a procedure different from $\algod$.
\suppress{
Indeed, the overlap between the $U$-flip of $\ket{\psi}$ and $\ket{\mu}$ is the same as the one
between $\ket{\psi}$ and $\ket{\mu}$.}
One possibility is to try to use $U$ in order to move $\ket{\psi}$ to
an orthogonal state that is closer to~$\ket{\mu}$.
\begin{definition}
The {\em $U$-rotation} of $\ket{\psi}$ is defined as
$  \delta_0\ket{w_0} +
    \sum_{j} \delta_j(\ket{w_j^+}-\ket{w_{j}^-}) +\delta_{-1}\ket{w_{-1}}
$,
where the decomposition of~$\ket{\psi}$ in terms of the orthonormal set
of eigenvectors~$\set{\ket{w_j^+},\ket{w_{j}^-}}$
of~$U$ is given by Eq.~(\ref{eqn-decomp}).
\end{definition}
\begin{fact}
If $\ket{\psi}$ is
orthogonal to both the $1$-eigenspace and the $(-1)$-eigenspace of $U$,
then the $U$-rotation of $\ket{\psi}$ is orthogonal to $\ket{\psi}$.
\end{fact}

This operation can be implemented efficiently by the following procedure
with further assumptions on~$U$. Namely, we would like $U$ 
to avoid having any eigenvalue close to $-1$.
This is naturally the case when we consider the quantum analogue of a
reversible Markov chain whose eigenvalues are all positive.
\encadre{
$\algof(U,\Delta,\eps)$ \quad --- \quad
Input: $\ket{\psi}$
\begin{enumerate}
\item\label{ss1}
Apply $\Theta(\log(1/\eps))$ times
the phase estimation circuit $\mathbf{Estimate}$ for $U$ with precision $\Delta$
to the same state $\ket{\psi}$.
\item If the majority of estimated phases are negative\\
\null\quad Perform a Phase a Flip.\\
Otherwise do nothing.
\item Undo the Phase Estimations of Step~\ref{ss1}.
\end{enumerate}
}
\begin{theorem}\label{finding}
Assume that
all eigenvalues $e^{\complexi\alpha}$ of $U$ satisfy $\abs{\alpha}\leq\pi/2$.
Then for every $T\geq \qht_\eps(U_2,\ket{\mu})$, 
the procedure $\algof(U,1/T,\eps)$ maps $\ket{\tphi_0}$ to a state 
at Euclidean distance
$\Order(\sqrt{\eps}\;)$ 
from the $U$-rotation of $\ket{\tphi_0}$.
Moreover, the number of applications of $\text{c-}U$ 
is
$ \Order(\log(1/\eps)\times\qht_\eps(U_2,\ket{\mu}))$.
\end{theorem}
\begin{proof}
The proof follows the same argument as in Theorem~\ref{detecting}.
\end{proof}

\subsection{Comparison between $\qht$ and $\qht_\eps$}
\label{sec-comparison}

We assume henceforth that $\braket{\phi_0}{\mu}\neq 0$; otherwise, the problem
is trivial---by our definition $\qht_{\eps}(U_2, \ket{\mu}) =
\qht(U_2, \ket{\mu}) = 0$.

The Markov inequality immediately implies the following relationship
between these two measures:
\begin{fact}\label{factqcomparison}
For all $U_2$, $\ket{\mu}$, and $\eps$,
\[
\qht_{\eps}(U_2, \ket{\mu}) \quad \leq \quad \frac{1}{\eps}\; \qht(U_2,
\ket{\mu}).
\]
\end{fact}
The other direction requires a closer look at the spectral
decomposition of $U_2U_1$. In this section, we again follow the
framework of the abstract search algorithm $U=U_2U_1$.  The
eigenvalues of $U_2$ different from $1$ are either $-1$ or they appear
as complex conjugates $\e^{\pm \complexi \theta_j}$, where~$\theta_j
\in (0,\pi)$. For convenience, we assume that
$\theta_{-1}=\pi$, $0 = \theta_0 \leq
\theta_1 \leq \theta_2 \leq \ldots$,  
and we always use index $j$ for
positive integers.  
Recall that both $\ket{\mu}$ and the $1$-eigenvector $\ket{\phi_0}$ of $U_2$
have real entries.
Writing $\ket{\mu}$ in the eigenspace
decomposition of $U_2$ we get
\[
\ket{\mu} \quad = \quad 
    a_0 \ket{\phi_0} + \sum_j a_j \left( \ket{\phi_j^+}
    + \ket{\phi_j^-} \right) + a_{-1} \ket{\phi_{-1}},
\]
where $\ket{\phi_{-1}}$ is a $(-1)$-eigenvector of $U_2$,
and $\ket{\phi_j^\pm}$ are $\e^{\pm \complexi \theta_j}$-eigenvectors
of $U_2$ such that $a_{-1}$ and $a_j$ are real,
$\ket{\phi_{-1}}$ has real entries,
and $\ket{\phi_j^-}=\overline{\ket{\phi_j^+}}$.
Since $\ket{\phi_0}$ has real entries, $a_0$ is also real.

Ambainis \etal~\cite{AKR} (see also Tulsi~\cite{Tulsi08}) show the
following relation between the spectrum of $U_2$ and that of $U$.
\begin{lemma}[\cite{AKR,Tulsi08}]
\label{lemma-akr-tulsi}
The eigenvalues~$\e^{\pm \complexi \alpha}$
of the operator~$U$ are solutions of
the equation:
\[ a_0^2 \cot\frac{\alpha}{2} 
             + \sum_j a_j^2 \left(\cot \left(\frac{\alpha + \theta_j}{2}
               \right) + \cot\left( \frac{\alpha - \theta_j}{2} \right)
               \right)
             - a_{-1}^2 \tan \frac{\alpha}{2} \quad = \quad 0 .
\]             
\suppress{\[
\abs{\sin\frac{\alpha}{2}} \quad = \quad
    \frac{1}{2} \times
    \frac{a_0}{\sqrt{\sum_j \frac{a_j^2}{\cos\alpha-\cos \theta_j }
    + \frac{a_{-1}^2}{4\cos^2(\alpha/2)}}}.
\]}
The corresponding unnormalized eigenvectors $\ket{w_\alpha}
= \ket{\mu} + \complexi\ket{w_\alpha'}$
satisfy $\braket{\mu}{w'_\alpha}=0$
and
\[
\ket{w_\alpha'} \quad = \quad
    a_0 \cot\frac{\alpha}{2} \, \ket{\phi_0} 
    + \sum_j a_j  \left( \cot\left(\frac{\alpha - \theta_j}{2}\right)
      \ket{\phi_j^{+}} + \cot\left(\frac{\alpha + \theta_j}{2}\right)
      \ket{\phi_j^{-}} \right) 
    - a_{-1} \tan\frac{\alpha}{2} \, \ket{\phi}.
\]
\end{lemma}


As in the classical case, we are only able to upper bound $\qht$ by
$\qht_\eps$ for some particular target states~$\ket{\mu}$.  Therefore,
we consider in the following lemma (proof in
Appendix~\ref{apptransitive}) an arbitrary set of orthonormal vectors
$M=\{\ket{\mu_z}\}$ whose span contains $\ket{\phi_0}$.  In the case
of the quantum analogue of a Markov chain $P$
as in Definition~\ref{def-qmc}, a natural choice for $\ket{\mu_z}$ is
$\ket{z}\ket{p_z}$ for some~$z$ in the state space of the Markov
chain.  Recall that~$\ket{\tphi_0} = \ket{\phi_0} - a_0 \ket{\mu}$.
\begin{lemma}\label{transitive}
  Let $M=\{\ket{\mu_z}\}$ be a set of orthonormal vectors with real
  coefficients in the standard basis, such that
  $\ket{\phi_0}\in\Span(M)$.  For every $z$, let $\alpha_z$ be the smallest
  positive real number $\alpha$ such that $\e^{\pm \complexi \alpha}$
  are eigenvalues of the operator $U = U_2 (I-2\proj{\mu_z})$.  Then
  there exists $z$ such that the length of the projection of
  $\ket{\tphi_0}$ onto the subspace generated by $\ket{w_{\alpha_z}}$
  and $\ket{w_{- \alpha_z}}$ is at least $1/\sqrt{2}$.
\end{lemma}

\begin{corollary}\label{marioq}
  Let $M=\{\ket{\mu_z}\}$ be a set of real orthonormal vectors such
  that $\ket{\phi_0}\in\Span(M)$.  For all $U_2$ there exists $z$ such
  that for all $\eps \leq 1/2$,
$$
\qht_{\eps}( U_2, \ket{\mu_z}) \quad = \quad \frac{1}{\alpha_z}.
$$
\end{corollary}
\begin{theorem}\label{qcomparison} 
  Let $M=\{\ket{\mu_z}\}$ be a set of real orthonormal vectors such
  that $\ket{\phi_0}\in\Span(M)$.  For all $U_2$ there exists $z$ such
  that for all $\eps \leq 1/2$,
$$ 
\qht( U_2, \ket{\mu_z}) \quad \leq \quad
\qht_{\eps}( U_2, \ket{\mu_z}).
$$
\end{theorem}
\begin{proof}
This is a consequence of Corollary~\ref{marioq} since 
$\qht( U_2 , \ket{\mu_z})$ is by definition at most 
$\frac{1}{\alpha_z}$.
\end{proof}

\subsection{Quadratic detection speedup for reversible chains}\label{sec-quadratic}
Let $P$ be an ergodic
Markov chain over state space $X = \{1, \ldots , n\}$.  We further
suppose that $P$ is a reversible Markov chain with {\em positive\/}
eigenvalues, otherwise we simply replace $P$ with
$\gamma P+(1-\gamma)I$, for any $\gamma<1/2$.
Let $z \in X$.
\begin{theorem}
\label{thm-qht-sqrt}
Assume that the eigenvalues of $P$ are all positive.
Then we have the following relations:
\begin{itemize}\reduceitem
\item For all $z$, \ \ \  $\qht (P,z) \leq \sqrt{\eht (P,z)/2}$.
\item For all $z$ and $\eps$, \ \ \  $\qht_{\eps}(P,z) = \sqrt{\eht_{\eps}(P,z)}$.\label{s2}
\end{itemize}

\end{theorem}
\begin{proof}
We follow the notation introduced in Sections~\ref{classical}
and~\ref{ct}. 
Then 
$\ket{\phi_0} = \sum_{x} \sqrt{\pi_x} \, \ket{x}\ket{p_x}$,
$\ket{\mu}=\ket{z}\ket{p_z}$,
$\ket{\tphi_0} =  \sum_{x \in X \setminus \{z\}} \sqrt{\pi_x} \, \ket{x}\ket{p_x}$.
Let $\sqrt{\pi_{-z}} = \sum_j \nu_j v_j$ be the decomposition of $\sqrt{\pi_{-z}}$ in the
normalized eigenbasis of $P_{-z}$ where the eigenvalue
of $v_j$ is $ \cos \theta_j$, with 
$0 < \theta_1 \leq \ldots  \leq \theta_{n-1} < \pi/2$. 
Let $v_j[x]$ denote the $x$-coordinate of the vector $v_j$. 
We set $\ket{\xi_j} = \sum_{x \neq z} v_j[x] \ket{x}\ket{p_x}$ and
$\ket{\zeta_j} = \sum_{y \neq z} v_j[y] \ket{p_y^*}\ket{y}$.
Then $\ket{\tphi_0} = \sum_{j} \nu_j \ket{\xi_j}$.
For every $j$, the vectors 
$\ket{\xi_j}$ and $\ket{\zeta_j}$ generate an eigenspace of $W(P,z)$ that is also generated
by two normalized eigenvectors with eigenvalues respectively 
$\e^{2 \complexi \theta_j}$ and $\e^{-2 \complexi \theta_j}$. 
This argument is still true for $\mathsf{SWAP}\cdot\reflex(\ay_{-z})$
when we divide the phases by $2$, leading to eigenvalues
$\e^{ \complexi \theta_j}$ and $\e^{- \complexi \theta_j}$ (cf.\ Fact~\ref{fact-akr-sze}). 
Since the length of the projection of
$\ket{\tphi_0}$ to this eigenspace is $\nu_j^2$, 
we have 
$\qht (P,z) = \sum_{j=1}^{n-1} \frac{\nu_j^2}{\theta_j} = \E[\sqrt{H_z}]$. 

By the Jensen inequality we get
$$
\qht (P,z) \quad \leq \quad \sqrt{\E[H_z]} \quad \leq \quad \sqrt{\eht (P,z)/2}.
$$
The second relation above immediately follows from $\mathit{QH}^2 =
H_z$.
\end{proof}
The same quadratic speed-up as above holds when there are multiple
marked elements in the state space~$X$. The search algorithm and its
analysis are similar and are omitted from this article.

\subsection{On the quadratic speedup threshold}
In this section we consider a
broad class of quantum walks defined on undirected
graphs. We are able to show that for a special case of walks on graphs,
the quadratic speedup is tight.

\suppress{
Our ideal goal was to prove:
\begin{quote}\em
For any reversible quantum walk $U_2$ on an undirected graph $G$, 
there exists a reversible ergodic Markov chain $P$ on $G$,
such that for all $z$ and $\eps$,
$$\qht_\eps(U_2,\ket{z}\ket{\phi^z}) \geq \qht_\eps(P,z)
\qquad\text{(up to some multiplicative constant),}$$
or equivalently (because of Theorem~\ref{thm-qht-sqrt}) that
$$\qht_\eps(U_2,\ket{z}\ket{\phi^z}) \geq\sqrt{\eht_\eps(P,z)}
\qquad\text{(up to some multiplicative constant),}$$
where the $1$-eigenvector $\ket{\phi_0}$ of $U_2$ decomposes in
$\ket{\phi_0}=\sum_z \sqrt{\pi_z}\ket{z}\ket{\phi^z}$,
$\pi_z\geq 0$, and $\ket{\phi^z}$ is a unit vector.
\end{quote}
}

Let $X=\{1,2,\ldots,n\}$.  Our notion of quantum walk can be seen as a
walk on the edges of a given undirected graph $G(X,E)$.  Let
$\mathcal{H}(G)=\Span(\ket{xy}:(x,y)\in E)$ be the Hilbert space that
a quantum walk on $G$ should preserve.  In the rest of this section,
we only consider operators and states in $\mathcal{H}(G)$ for some
given $G$.  Observe that $\mathsf{SWAP}$ preserves $\mathcal{H}(G)$
since $G$ is undirected.

We introduce a notion of reversibility that is justified by
Lemma~\ref{classical-analogue} below.
\begin{definition}
A {\em quantum walk} on an undirected graph $G=(X,E)$ is
a unitary $U_2=\mathsf{SWAP}\cdot F$ defined on a subspace of $\mathcal{H}(G)$, 
where $F$ is matrix with real entries of the form
$F=\sum_{x\in X}\proj{x}\otimes F^x$,
and where $U_2$ has a single $1$-eigenvector $\ket{\phi_0}$.
The quantum walk is {\em reversible} when $\mathsf{SWAP}(\ket{\phi_0})=\ket{\phi_0}$.
\end{definition}
Observe that the definition implies that $\ket{\phi_0}$ can be chosen with real entries.
This definition of quantum walk appears, for example, in the survey paper of 
Ambainis~\cite{amb}, see also~\cite{sant}.
Szegedy considered for $F^x$ a specific kind of reflection based 
on Markov chain transition probabilities (see Section~\ref{ct}).
\begin{definition}\label{classicalization}
Let $U_2=\mathsf{SWAP}\cdot F$ be a quantum walk with unit $1$-eigenvector
$\ket{\phi_0}=\sum_x \sqrt{\pi_x}\ket{x}\ket{\phi^x}$, where $\pi_x\geq 0$
and $\ket{\phi^x}$ is a unit vector with real entries.
Then the {\em classical analogue}~$P = (p_{xy})$ of $U_2$ is defined as
$p_{xy}=\braket{y}{\phi^x}^2$.
\end{definition}
Since $\ket{\phi_0}$ is a $1$-eigenvector of $\mathsf{SWAP}\cdot F$ we
directly state:
\begin{fact}\label{one}
Let $\ket{\psi^x}=F^x\ket{\phi^x}$.
Then
$\ket{\phi_0}=\sum_x\sqrt{\pi_x}\ket{\psi^x}\ket{x}$.
\end{fact}

\begin{lemma}\label{classical-analogue}
The classical analogue $P$ of a quantum walk $U_2$ on $G$
is a Markov chain on $G$ with stationary probability distribution $\pi$.
Moreover, $P$ is reversible if and only if $U_2$ is a reversible quantum walk.
\end{lemma}
\begin{proof}
First we show that $P$ is a Markov chain on $G$.
For every $x$, we have
$$\sum_y p_{xy}=\sum_y\braket{y}{\phi^x}^2=\norm{\phi^x}^2=1.$$
Moreover $p_{xy}\neq 0$ implies $\braket{y}{\phi^x}\neq 0$,
which implies that $(x,y)\in E$ 
since $\ket{\phi_0}\in \mathcal{H}(G)$.

Now we verify that $\pi$ is a stationary probability distribution.
First, $\pi$ is a probability distribution since $\ket{\phi^x}$ for
all $x \in X$ and $\ket{\phi_0}$ are unit vectors.
That $\pi$ is a stationary probability distribution can be seen from
the following sequence of equalities which hold for every $y \in X$:
\begin{eqnarray*}
\sum_x\pi_xp_{xy}&=&\sum_x\braket{xy}{\phi_0}^2\qquad\text{by definition of $P$ 
and $\ket{\phi_0}$}\\
&=&\sum_x\pi_y\braket{x}{\psi^y}^2\qquad\text{by Fact~\ref{one}}\\
&=&\pi_y\norm{\ket{\psi^y}}^2=\pi_y.
\end{eqnarray*}

For reversibility, observe that we similarly have
$\pi_x p_{xy}=\braket{xy}{\phi_0}^2$ and 
$\pi_y p_{yx}=\braket{yx}{\phi_0}^2=(\bra{xy}\mathsf{SWAP}\ket{\phi_0})^2$. 
$P$ is reversible when these two expressions are equal for every $x,y$,
which happens precisely when the quantum walk $U_2$ is reversible.
\end{proof}

Finally, we show that the quadratic speedup is tight in the special
case of walks for which all of the unitaries $F^x$ are reflections.
We state the result using the notation above.
\begin{theorem}\label{lowerbound}
Let $U_2=\mathsf{SWAP}\cdot F$ be a reversible quantum walk such that 
$F^x=2\proj{\phi^x}-I$, for all $x\in X$.
Then for all $z$ and $\eps$,
$$\qht_\eps(U_2,\ket{z}\ket{\phi^z}) = \qht_\eps(P,z)=\sqrt{\eht_\eps(P,z)}.$$
\end{theorem}
\begin{proof}
Let $U_1=I-2\proj{z}\otimes\proj{\phi^z}$, for some fixed $z$.
Under the hypothesis of the theorem, $(U_2U_1)^2$ is a product of two reflections
$\reflex(\ay_{-z})$ and $\reflex(\bee_{-z})$, where
$\ay_{-z}=\Span(\ket{x}\ket{\phi^x}:x\in X\setminus\{z\})$
and $\bee_{-z}=\Span(\ket{\phi^y}\ket{y}:y\in X\setminus\{z\})=\mathsf{SWAP}(\ay)$.

{From}~\cite{Sze}, we know that the spectrum of $(U_2U_1)^2$ is
completely defined by the discriminant matrix $D=A^*B$, where
$A=\sum_{x\neq z}\ket{x}\ket{\phi^x}\bra{x}$ and $B=\sum_{y\neq
  z}\ket{\phi^y}\ket{y}\bra{y}$.  We get that
$D=(\braket{x}{\phi^y}\braket{\phi^x}{y})_{x\neq z,y\neq z}$.  The
reversibility of $U_2$ guarantees that
$\braket{xy}{\pi}=\braket{yx}{\pi}$, which implies that
$\sqrt{\pi_x}\braket{y}{\phi_x}=\sqrt{\pi_y}\braket{x}{\phi_y}$.
Since $\ket{\phi_y}$ has real entries, we have
$D=\sqrt{\Pi}P_{-z}\sqrt{\Pi}^{-1}$, where ${\Pi}=\diag(\pi_x)_{x\neq
  z}$ and $P_{-z}$ is the matrix obtained from $P$ by deleting the row
and column indexed by $z$.

Observe that this discriminant is exactly that of the quantum analogue
$W(P,z)$. So $W(P,z)$ and $(U_2U_1)^2$ are equal up to a basis change
which maps $\sum_x\sqrt{\pi_x}\ket{x}\ket{p_x}$ to $\ket{\phi_0}$,
$\ket{z}\ket{p_z}$ to $\ket{z}\ket{\phi^z}$, and therefore
$\sum_{x\neq z}\sqrt{\pi_x}\ket{x}\ket{p_x}$ to $\ket{\tphi_0}$.
\end{proof}

\section{Finding with constant probability}
\label{find}

In this section, we extend a technique devised by Tulsi~\cite{Tulsi08} for
finding a marked state on the 2D grid in time that is the
square-root of the classical hitting time.
We prove that it may be applied to a larger class of Markov chains and
target states. The technique may be combined with ideas developed in the earlier sections
to give an algorithm for the quantum analogue of
an arbitrary reversible ergodic Markov chain.

We use the notation of Section~\ref{ct}.
In our application, there is no $(-1)$-eigenvector of $U_2$.
Therefore the marked state $\ket{\mu}$ has the following
decomposition in an eigenvector basis of~$U_2$:
\begin{equation}
\label{eqn-mu}
\ket{\mu} \quad = \quad a_0 \ket{\phi_0} + \sum_{1 \le j \le J} a_j
(\ket{\phi_j^+} + \ket{\phi_j^-}),
\end{equation}
where~$J$ is some positive integer.
Last, we assume in the rest of this section
that $\braket{\phi_0}{\mu}\neq 0$.

\begin{lemma}
\label{thm-ev-rep}
The vectors~$\ket{\mu}$ and~$\ket{\tphi_0}$ have the following
representation in the basis~$\set{\ket{w_\alpha}}$ consisting of the
eigenvectors of~$U = U_2 U_1$ as given by
Lemma~\ref{lemma-akr-tulsi}:
$\ket{\mu}  =  \sum_{\alpha} \frac{1}{\norm{w_\alpha}^2} \ket{w_\alpha}$,
and 
$\ket{\tphi_0}  =  \sum_{\alpha} \frac{a_0 \complexi\cot(\frac{\alpha}{2})}
                    {\norm{w_\alpha}^2} \ket{w_\alpha}$.
\end{lemma}
\begin{proof}
Any vector~$\ket{\psi}$ may be expressed in the orthogonal
basis~$\set{\ket{w_\alpha}}$ as
$
\ket{\psi} = \sum_\alpha \frac{\braket{w_\alpha}{\psi}}
                         {\norm{w_\alpha}^2} \ket{w_\alpha}$.
The first equation now follows from
$\braket{w_{\alpha}}{\mu} = (\bra{\mu}-\complexi \bra{w'_{\alpha}})
\ket{\mu} = 1$.

By Lemma~\ref{lemma-akr-tulsi},
$
\braket{\phi_0}{w_{\alpha}} = \bra{\phi_0}(\ket{\mu} +
\complexi \ket{w'_{\alpha}}) = a_0 + a_0 
\complexi\cot \frac{\alpha}{2}$.
The second equation follows by combining the above
with~$\ket{\tphi_0} = \ket{\phi_0} - a_0 \ket{\mu}$.
\end{proof}

\begin{lemma}\label{qht-obs}
\suppress{
The~$\ket{\mu}$-hitting time of~$U_2$ is then given by
\[
\qht(U_2,\ket{\mu}) \quad = \quad
\sum_{\alpha > 0}  \frac{a_0^2 \cot^2(\frac{\alpha}{2})}{\norm{w_\alpha}^2}
\times \frac{2}{\alpha},
\]
and }
The inner product of the target state~$\ket{\mu}$ and the~$U$-rotation
of~$\ket{\tphi_0}$ 
is $
2 a_0 \sum_{\alpha > 0}  \frac{\cot(\frac{\alpha}{2})}{\norm{w_\alpha}^2}
$.
\end{lemma}
Theorem~\ref{thm-qht-sqrt} shows that the quantum hitting time is
bounded by the square-root of the classical hitting time when~$U_2$ is
derived from a reversible Markov chain~$P$, i.e., $U_2 =
\mathsf{SWAP}\cdot\reflex(\ay)$ in the notation of
Section~\ref{sec-algorithms}. This allows for the detection of marked
elements (or more generally for checking if~$\braket{\mu}{\phi_0} \neq
0$) and also the creation of the~$U$-rotation of~$\ket{\tphi_0}$
in the stated time bound.  However, the overlap of the
$U$-rotation of~$\ket{\tphi_0}$ with the target $\ket{\mu}$ may
be~$\order(1)$. Tulsi~\cite{Tulsi08} discovered a technique, described
below, to boost this overlap to~$\Omega(1)$ in the case of a quantum
walk on the 2D grid.

Let $\theta\in[0,\pi/2)$. Let~$R_\theta$ denote the rotation
in~$\complex^2$ by angle~$\theta$:
\[
R_\theta \quad = \quad \left[ \begin{array}{cc}
                                 \cos\theta & -\sin\theta \\
                                 \sin\theta & \cos\theta
                              \end{array}
                       \right],
\]
and let~$\ket{\theta} = R_\theta^\dagger\ket{0}$,
and~$\ket{\theta^\perp} = R_\theta^\dagger\ket{1}$. 
Define 
$U_1^\theta=\density{0} \tensor \id + \density{1} \tensor U_1$, and
$U_2^\theta=( \density{\theta} \tensor (-\id)
                    +\density{\theta^\perp} \tensor U_2 )$.
Then  $U_1^\theta=\id - 2 \proj{1}\otimes\proj{\mu}$, meaning that
the modified marked state is $\ket{1}\ket{\mu}$.
Then the modified abstract search algorithm becomes:
\begin{eqnarray}
\label{eqn-tulsi-def}
\tulsi(U_1,U_2,\theta) 
    \quad = \quad U_2^\theta U_1^\theta
    \quad = \quad ( \density{\theta} \tensor (-\id )
                    +\density{\theta^\perp} \tensor U_2 )
                  ( \density{0} \tensor \id + \density{1} \tensor U_1 )
\end{eqnarray}
This is precisely the circuit used by Tulsi: his
rotation~$\hat{R}_{\theta}=R_{\theta}^\dagger$ in our notation.
Tulsi~\cite{Tulsi08} proved that the principal eigenvalue
of the operator above is
closely related to that of the unitary operator~$U_2 U_1$.
We extend his findings in terms that are more readily used in our
context.

The eigenvalues of~$U_2^\theta$ are the same as those of~$U_2$, except
for the addition of the new eigenvalue $-1$.
The eigenvectors corresponding to eigenvalues~$\e^{\pm \complexi \theta_j}$
are now~$\ket{\theta^\perp}\ket{\phi_j^{\pm}}$.
Any state of the form~$\ket{\theta}\ket{\psi}$ is a~$-1$ eigenvector of~$U_2^\theta$.
\begin{fact}
The decomposition of $\ket{1}\ket{\mu}$ in the eigenvector basis
of~$U_2^\theta$ is:
\[
\ket{1}\ket{\mu} \quad = \quad
    \cos\theta \; \ket{\theta^\perp} \left(a_0 \ket{\phi_0} 
    + \sum_{1 \le j \le J} a_j (\ket{\phi_j^{+}} + \ket{\phi_j^{-}})
    \right) - \sin\theta \; \ket{\theta}\ket{\mu},
\]
where the coefficients~$a_0,a_j$ are precisely those in
Eq.~(\ref{eqn-mu}).
\end{fact}

\begin{lemma}\label{cor-tulsi}
The eigenvalues~$\e^{\pm \complexi \alpha^\theta}$,
of the operator~$\tulsi(U_1,U_2,\theta)$ are solutions to
the equation
\[
a_0^2 \cot\frac{x}{2} 
             + \sum_j a_j^2 \left(\cot \left(\frac{x + \theta_j}{2}
               \right) + \cot\left( \frac{x - \theta_j}{2} \right)
               \right)
             - \tan^2\theta \tan \frac{x}{2} \quad = \quad 0.
\]
\suppress{or equivalently
\[
\abs{\sin \left(\frac{\alpha^\theta}{2}\right)}
\quad = \quad \frac{1}{2} \times \frac{a_0}{\sqrt{\sum_j
\frac{a_j^2}{\cos\alpha^\theta-\cos \theta_j}+ 
\frac{\tan^2\theta}{4 \cos^2({\alpha^\theta}/{2})}}}.
\]}
The corresponding unnormalized eigenvectors $\ket{w_{\alpha,\theta}}
= \ket{1}\ket{\mu} + \complexi\ket{w'_{\alpha,\theta}}$ satisfy
$\braket{1,\mu}{w'_{\alpha,\theta}}=0$
and
\begin{eqnarray*}
\ket{w'_{\alpha,\theta}} &=&
 \cos\theta\; \ket{\theta^\perp}\left(a_0 \cot\left(\frac{\alpha^\theta}{2}\right) \, \ket{\phi_0}
+ \sum_j a_j  \left( \cot\left(\frac{\alpha^\theta - \theta_j}{2}\right)
\ket{\phi_j^{+}} + \cot\left(\frac{\alpha^\theta + \theta_j}{2}\right) 
\ket{\phi_j^{-}} \right)\right) \\
&& \mbox{} + \sin\theta\; \ket{\theta} \left(
\tan\left(\frac{\alpha^\theta}{2}\right) \ket{\mu}\right).
\end{eqnarray*}
\end{lemma}
\begin{proof}
We apply Lemma~\ref{lemma-akr-tulsi}
from Section~\ref{sec-comparison} with $a_0^\theta = a_0 \cos
\theta$, $a_j^\theta = a_j \cos \theta$, $a_{-1}^\theta = \sin \theta$.
Note that $U_2$ does not have any $(-1)$-eigenvectors (by assumption), but
$U_2^\theta$ does.
\end{proof}

The target vector in the modified algorithm is~$\ket{1}\ket{\mu}$. The
start state is chosen to be~$\ket{\tphi_{0,\theta}} =
\ket{\theta^\perp}\ket{\tphi_0}$. The following are analogous to
Lemmata~\ref{thm-ev-rep} and~\ref{qht-obs}:
\begin{corollary}
\label{cor-ev-decomp-tulsi}
The vectors~$\ket{1}\ket{\mu}$ and~$\ket{\tphi_{0,\theta}}$ have the following
representation in the basis~$\set{\ket{w_{\alpha,\theta}}}$ consisting of the
eigenvectors of~$\tulsi(U_1, U_2, \theta)$ as given by
Corollary~\ref{cor-tulsi}:
$$\ket{1}\ket{\mu}  =  \sum_{\alpha^\theta} \frac{1}{\norm{w_{\alpha,\theta}}^2}
                  \ket{w_{\alpha,\theta}}, \quad\text{and}\quad
\ket{\tphi_{0,\theta}}  =  (a_0 \complexi\, \cos\theta) \sum_{\alpha^\theta}
                  \frac{\cot(\frac{\alpha^\theta}{2})}
                    {\norm{w_{\alpha,\theta}}^2} \ket{w_{\alpha,\theta}}.$$
\end{corollary}
\begin{corollary}
\label{cor-qht-tulsi}
\suppress{The~$\ket{1}\ket{\mu}$-hitting time of~$U_2^\theta$ is then given by
\[
\qht(U_2^\theta,\ket{1}\ket{\mu}) \quad = \quad
(a_0 \cos\theta)^2 \sum_{\alpha^\theta > 0} 
     \frac{\cot^2(\frac{\alpha^\theta}{2})}{\norm{w_{\alpha,\theta}}^2}
\times \frac{2}{\alpha^\theta}, 
\]
and} 
The inner product of the target state $\ket{1}\ket{\mu}$
and the~$\tulsi(U_1,U_2,\theta)$-rotation of~$\ket{\tphi_{0,\theta}}$
is given by the expression
$
(2 a_0 \cos\theta) \sum_{\alpha^\theta > 0}
 \frac{\cot(\frac{\alpha^\theta}{2})}{\norm{w_{\alpha,\theta}}^2}$.
\end{corollary}

We choose for the rest of this section $\theta \in [0,\pi/2]$ such that
$\tan\theta =
{a_0 \cot ({\alpha_1}/{2})}/10
$.
Let~$\alpha^\theta_1$ be the smallest positive eigenphase of 
the modified search algorithm~$\tulsi(U_1,U_2,\theta)$.

Lemma~\ref{lem-approximation1} 
(proof in Appendix~\ref{applem-approximation1})
proves that~$\alpha^\theta_1$ 
is of the same order as the principal eigenphase~$\alpha_1$ of the 
original algorithm~$U_2 U_1$.
Lemma~\ref{lem-approximation2} (proof in Appendix~\ref{applem-approximation2})
is the final piece in our argument. It relates the norm of the principal 
eigenvectors of the modified walk to the norm of the original ones.
Both lemmas extend corresponding results by Tulsi
in the case of the 2D grid.
\begin{lemma}
\label{lem-approximation1}
There is a unique eigenvalue phase $\alpha_1^\theta$ of the
operator~$\tulsi(U_1,U_2,\theta)$ 
in $(0,\alpha_1]$.
Moreover, 
$\cot(\alpha_1^\theta/2)\leq 1.01\times \cot(\alpha_1/2)$.
Therefore  if $0\leq\alpha_1\leq\pi/4$, then
$0.78\times\alpha_1\leq\alpha_1^\theta\leq\alpha_1$.
\end{lemma}

\begin{lemma}
\label{lem-approximation2}
$\norm{w_{\pm\alpha_1,\theta}} \leq
(3 \cos\theta)\times \norm{w_{\pm\alpha_1}}$.
\end{lemma}

We have all the ingredients for the main result of this section.
\begin{theorem}
Let $\eps>0$ be any constant.
Suppose that the squared length
of the projection of the state~$\ket{\tphi_0}$ onto the principal
eigenspace of~$U_2 U_1$ is bounded below by $1-\eps$. 
Then, for every $T\geq \qht_\eps(U_2,\ket{\tphi_{0}})/0.78$,
the procedure $\algof(\tulsi(U_1,U_2,\theta),1/T,1/4)$ maps $\ket{\tphi_{0,\theta}}$ 
to a state with constant overlap with the target state~$\ket{1}\ket{\mu}$.
\end{theorem}
\begin{proof}
First we prove that
$T=\qht_\eps(U_2^\theta,\ket{1}\ket{\mu})$ is of the order of~$\qht_\eps(U_2,\ket{\mu})$.
Let $l = 2 a_0^2 (\cot^2 \tfrac{\alpha_1}{2}) / \norm{w_{\alpha_1}}^2$.
We know that $l\geq 1-\eps$. Using Lemma~\ref{thm-ev-rep} we get that
$\qht_\eps(U_2,\ket{\mu}) = 1/\alpha_1$.
Moreover, by definition, $T \leq 1/\alpha_1^\theta$. 
By Lemma~\ref{lem-approximation1},
$T \leq 1/(0.78\alpha_1) = \qht_\eps(U_2,\ket{\mu}) / 0.78 $.
We now get our conclusion
by applying Corollary~\ref{cor-qht-tulsi}, Lemmata~\ref{lem-approximation1}
and~\ref{lem-approximation2}, and Theorem~\ref{finding}.
\end{proof}

We combine the above theorem with Lemma~\ref{transitive}
to get our final result. 
\begin{corollary}
  Let $P$ be a state-transitive reversible ergodic Markov chain, and
  let $z$ be any state.  Set $\ket{\mu}=\ket{z}\ket{p_z}$,
  $U_1=I-2\proj{\mu}$, and let $U_2$ be the quantum analogue of $P$.
 Then for every $\eps\leq 1/2$
 and $T\geq \qht_\eps(U_2,\ket{\tphi_{0}})/0.78$, the procedure
  $\algof(\tulsi(U_1,U_2,\theta),1/T,1/4)$ maps
  $\ket{\tphi_{0,\theta}}$ to a state with constant overlap with the
  target state~$\ket{1}\ket{\mu}$.
\end{corollary}
\begin{proof}
The proof is direct once we realize that
 the conclusions of Lemma~\ref{transitive} for a particular element $z$
 remain valid for any element because of the state-transitivity of $P$.
\end{proof}
\bibliography{qwalk}

\newcommand{\etalchar}[1]{$^{#1}$}
\begin{thebibliography}{CEMM98}

\bibitem[AA05]{AA}
S.~Aaronson and A.~Ambainis.
\newblock Quantum search of spatial regions.
\newblock {\em Theory of Computing}, 1(4):47--79, 2005.

\bibitem[AAKV01]{AAKV}
D.~Aharonov, A.~Ambainis, J.~Kempe, and U.~Vazirani.
\newblock Quantum walks on graphs.
\newblock {\em Proceedings of the 33rd Annual ACM Symposium on Theory of
  Computing}, pages 50--59, 2001.

\bibitem[ABN{\etalchar{+}}01]{ABNVW}
A.~Ambainis, E.~Bach, A.~Nayak, A.~Vishwanath, and J.~Watrous.
\newblock One-dimensional quantum walks.
\newblock {\em Proceedings of the 33rd Annual ACM Symposium on Theory of
  Computing}, pages 37--49, 2001.

\bibitem[AKR05]{AKR}
A.~Ambainis, J.~Kempe, and A.~Rivosh.
\newblock Coins make quantum walks faster.
\newblock In {\em Proceedings of the Sixteenth Annual ACM-SIAM Symposium on
  Discrete Algorithms}, pages 1099--1108, 2005.

\bibitem[Amb03]{amb}
A.~Ambainis.
\newblock Quantum walks and their algorithmic applications.
\newblock {\em International Journal of Quantum Information}, 1:507--518, 2003.

\bibitem[Amb04]{Amb1}
A.~Ambainis.
\newblock Quantum walk algorithm for {Element Distinctness}.
\newblock {\em Proceedings of the 45th Symposium on Foundations of Computer
  Science}, pages 22--31, 2004.

\bibitem[B{\v{S}}06]{BS}
H.~Buhrman and R.~{\v{S}}palek.
\newblock Quantum verification of matrix products.
\newblock {\em Proceedings of the Seventeenth Annual ACM-SIAM Symposium on
  Discrete Algorithms}, pages 880--889, 2006.

\bibitem[CEMM98]{CleveEMM98}
R.~Cleve, A.~Ekert, C.~Macchiavello, and M.~Mosca.
\newblock Quantum algorithms revisited.
\newblock {\em Proceedings of the Royal Society of London, Series A},
  454:339--354, 1998.

\bibitem[CG04]{CG2}
A.~Childs and J.~Goldstone.
\newblock Spatial search and the {D}irac equation.
\newblock {\em Physical Review A}, 70:042312, 2004.

\bibitem[Gro96]{Gro}
L.~Grover.
\newblock A fast quantum mechanical algorithm for database search.
\newblock {\em Proceedings of the Twenty-Eighth Annual ACM Symposium on the
  Theory of Computing}, pages 212--219, 1996.

\bibitem[Kit95]{Kitaev96}
A.~Kitaev.
\newblock Quantum measurements and the {A}belian stabilizer problem.
\newblock ECCC technical report 96-003 and arXiv.org e-print quant-ph/9511026,
  1995.

\bibitem[MN07]{MN}
F.~Magniez and A.~Nayak.
\newblock Quantum complexity of testing group commutativity.
\newblock {\em Algorithmica}, 48(3):221--232, 2007.

\bibitem[MNRS07]{MNRS}
F.~Magniez, A.~Nayak, J.~Roland, and M.~Santha.
\newblock Search via quantum walk.
\newblock {\em Proceedings of the 39th ACM Symposium on Theory of Computing},
  pages 575--584, 2007.

\bibitem[MSS07]{MSS}
F.~Magniez, M.~Santha, and M.~Szegedy.
\newblock Quantum algorithms for the triangle problem.
\newblock {\em SIAM Journal on Computing}, 37(2):611--629, 2007.

\bibitem[NV00]{NV}
A.~Nayak and A.~Vishwanath.
\newblock Quantum walk on the line.
\newblock Technical Report quant-ph/0010117, arXiv, 2000.

\bibitem[San08]{sant}
M.~Santha.
\newblock Quantum walk based search algorithms.
\newblock In {\em Proceedings of the 5th Annual Conference on Theory and
  Applications of Models of Computation}, volume 4978, pages 31--46. LNCS,
  2008.

\bibitem[SKW03]{SKW}
N.~Shenvi, J.~Kempe, and K.~Whaley.
\newblock A quantum random walk search algorithm.
\newblock {\em Physical Review A}, 67:052307, 2003.

\bibitem[Sze04]{Sze}
M.~Szegedy.
\newblock Quantum speed-up of {Markov} chain based algorithms.
\newblock {\em Proceedings of the 45th Symposium on Foundations of Computer
  Science}, pages 32--41, 2004.

\bibitem[Tul08]{Tulsi08}
A.~Tulsi.
\newblock Faster quantum walk algorithm for the two dimensional spatial search.
\newblock {\em Physical Review A}, 2008.
\newblock To appear.

\end{thebibliography}

\appendix

\section{Proof of Theorem~\ref{justif}}\label{appjustif}
\begin{proof}
  For the first statement, we set $k = \left(4 \ln
    \frac{2}{\eps}\right) \eht_{\eps/2}(P,z)$, and we denote by $s_k$
  the probability that $z$ is not reached in the first $k$ steps. The
  claim follows if we show that $s_k \leq \eps$. It is not hard to see
  that $s_k = \pi_{-z}^\adjoint P_{-z}^k u_{-z} = \sum_j \nu_j^2 (\cos
  \theta_j)^k$. We bound $s_k$ from above as follows:
\begin{eqnarray*}
s_k & \leq & \sum_{j : 1/ \theta_j^2 >  \eht_{\eps /2}(P,z) }\nu_j^2  ~+~ 
\sum_{j : 1/ \theta_j^2 \leq  \eht_{\eps /2}(P,z) }\nu_j^2 (\cos \theta_j)^k \\
  & \leq & \eps/2 ~+~ \left(1 - \frac{1}{4\; \eht_{\eps /2}(P,z)} \right)^k .
\end{eqnarray*}
The first summation is at most $\eps/2$ by the definition of
$\eht_{\eps /2}(P,z)$, and the bound for the second summation is
justified since $1/ \theta_j^2 \leq \eht_{\eps /2}(P,z)$ implies $\cos
\theta_j \leq 1 - 1/ (4 \eht_{\eps /2}(P,z))$.  Thus the second
term is also at most $\eps/2$ by the choice of $k$.

For the second statement, set $k = \frac{1}{2} h_{\eps/3}(P,z)$. Then
using the definition of $ h_{\eps/3}(P,z)$ and bounding $(1 -
1/2k)^{2k}$ from below by 1/3, we get
\[
\frac{\eps}{3} \quad \geq \quad
s_{2k} \quad = \quad
\sum_j  \nu_j^2 (\cos \theta_j)^{2k}  \quad > \quad
\sum_{j : \cos \theta_j > 1 - \frac{1}{2k}} \nu_j^2 
\left(1 - \frac{1}{2k} \right)^{2k}  \quad \geq \quad
\frac{1}{3} \sum_{j : \cos \theta_j > 1 - \frac{1}{2k}} \nu_j^2.
\]
Now observe that if $1/ \theta_j^2 > k$, then $\cos \theta_j > 1 -
\frac{1}{2k}$.  Therefore $\Pr [ H_z > k] \leq \eps$, and the
statement follows.
\end{proof}

\section{Proof of Lemma~\ref{transitive}}\label{apptransitive}
\begin{proof}
We know from Lemma~\ref{lemma-akr-tulsi} that if
\[
\ket{\mu_z} \quad = \quad 
    a_{0,z} \ket{\phi_0} + \sum_j a_{j,z} \left( \ket{\phi_j^+}
    + \ket{\phi_j^-} \right) + a_{-1,z} \ket{\phi},
\]
then
\[ a_{0,z}^2 \cot\frac{\alpha}{2} \quad  = \quad
             - \sum_j a_{j,z}^2 \left(\cot \left(\frac{\alpha + \theta_j}{2}
               \right) + \cot\left( \frac{\alpha - \theta_j}{2} \right)
               \right)
             + a_{-1,z}^2 \tan \frac{\alpha}{2},
\]             
where $0 < \alpha_z < \theta_1 \leq \theta_2 \leq \ldots < \pi$.
Since
$$- \left(\cot \left(\frac{\alpha + \theta_j}{2}
               \right) + \cot\left( \frac{\alpha - \theta_j}{2} \right)
               \right)
\quad = \quad \frac{2 \sin \alpha}{\cos \alpha - \cos \theta_j},$$
this is equivalent to
\begin{eqnarray}
\label{eqn-largest-ev}
a_{0,z}^2 \cot\frac{\alpha}{2} & = &
\sum_j a_{j,z}^2
\frac{2 \sin \alpha}{\cos \alpha - \cos \theta_j}
+ a_{-1,z}^2 \tan \frac{\alpha}{2}.
\end{eqnarray}
We first claim that there exists $z$ such that $|a_{0,z}|^2 \leq |a_{1,z}|^2$. Indeed, 
since $\ket{\phi_0}$ belongs to $\Span(M)$,
$$
\sum_{z} a_{0,z}^2 \quad = \quad \sum_{z} \braket{\mu_z}{\phi_0}
\quad = \quad 1 
\quad = \quad 
\sum_{z } \braket{\mu_z}{\phi_1^+} \quad = \quad \sum_{z } a_{1,z}^2.
$$
Fix now arbitrarily such a $z$. To simplify the notation, from now on we refer to $\alpha_z$
and to the coefficients $a_{j,z}$ without the subscript $z$.
Since all of the terms on the right hand side of Eq.~(\ref{eqn-largest-ev})
are positive, it follows that
\begin{equation}\label{eq1}
  \cot\frac{\alpha}{2} \quad \geq \quad 
  \frac{2 \sin \alpha}{\cos \alpha - \cos \theta_1}.
\end{equation}
Since the right hand side decreases if $\theta_1$ is replaced by some
$\theta_1 < \theta \leq \pi$, we obtain for every $j >1$,
\begin{equation}\label{eq2}
  \cot\frac{\alpha}{2} \quad \geq \quad 
  \frac{2 \sin \alpha}{\cos \alpha - \cos \theta_j},
\end{equation}
and
\begin{equation}\label{eq3}
  \cot\frac{\alpha}{2} \quad \geq \quad 
  \frac{2 \sin \alpha}{\cos \alpha - \cos \pi} 
  \quad = \quad 2 \tan \frac{\alpha}{2} \quad > \quad
  \tan \frac{\alpha}{2}.
\end{equation}

We also know that the 
eigenvectors $\ket{w_{\pm \alpha}}
= \ket{\mu} + \complexi\ket{w_{\pm \alpha}'}$
corresponding to the eigenvalues $\e^{\pm \complexi \alpha}$
satisfy 
\[
\ket{w_{ \pm \alpha}'} = 
    a_0 \cot\frac{\pm \alpha}{2} \, \ket{\phi_0} 
    + \sum_j a_j  \left( \cot\left(\frac{\pm \alpha - \theta_j}{2}\right)
      \ket{\phi_j^{+}} + \cot\left(\frac{\pm \alpha + \theta_j}{2}\right)
      \ket{\phi_j^{-}} \right) 
    - a_{-1} \tan\frac{\pm \alpha}{2} \, \ket{\phi}.
\]
Let us now define the vector $\ket{s}$ in the two dimensional space 
generated by
$\ket{w_{\pm \alpha}}$ by
$$
\ket{s} \quad = \quad \frac{ \ket{w_{ \alpha}} - \ket{w_{- \alpha}}}
{\complexi}.
$$
Observe that $\ket{\mu}$ is orthogonal to $\ket{w_{\pm\alpha}}$ and
therefore to $\ket{s}$.  Then the length of the projection of
$\ket{\tphi_0}$ to the subspace is the same as the one of
$\ket{\phi_0}$. This is at least 
$|\braket{s}{\phi_0}|/ \norm{s}$, which we now bound from below. Since
the functions $\tan$ and $\cot$ are odd, we get
$$\ket{s} \quad = \quad
2 a_0 \cot\frac{ \alpha}{2} \, \ket{\phi_0}  +
\sum_j a_j  \left( \cot\left(\frac{ \alpha - \theta_j}{2}\right)
      + \cot\left(\frac{ \alpha + \theta_j}{2}\right)  \right)
      \left( \ket{\phi_j^{+}} + \ket{\phi_j^{-}} \right) -
      2 a_{-1} \tan\frac{ \alpha}{2} \, \ket{\phi},
$$
and therefore 
$$
\norm{s}^2 \quad = \quad 
4 a_0^2  \cot^2\frac{ \alpha}{2} + 
8 \sum_j a_j^2 \frac{ \sin^2 \alpha}{(\cos \alpha - \cos \theta_j)^2} +
4 a_{-1}^2 \tan^2\frac{ \alpha}{2}.
$$
{From}~\eqref{eqn-largest-ev}, \eqref{eq1}, \eqref{eq2}, and~\eqref{eq3} it follows that
$$ 
a_0^2  \cot^2\frac{ \alpha}{2} \quad \geq \quad
2\sum_j a_j^2 \frac{ \sin^2 \alpha}{(\cos \alpha - \cos \theta_j)^2} +
a_{-1}^2 \tan^2\frac{ \alpha}{2}
$$
and therefore
$$
\norm{s}^2 \quad \leq \quad 8 a_0^2  \cot^2\frac{ \alpha}{2}.
$$
Since $
\braket{s}{\phi_0} = 2 a_0 \cot\frac{ \alpha}{2}$, 
we can indeed conclude that
$$\frac{|\braket{s}{\phi_0}|}
{\norm{s}} \quad \geq \quad \frac{1}{\sqrt{2}}.
$$
\end{proof}

\section{Proof of Lemma~\ref{lem-approximation1}}\label{applem-approximation1}
\begin{proof}
By the definition of $\alpha_1$ (Lemma~\ref{lemma-akr-tulsi}),
\begin{equation}\label{alpha}
a_0^2 \cot\frac{\alpha_1}{2} 
             + \sum_j a_j^2 \left(\cot \left(\frac{\alpha_1 + \theta_j}{2}
               \right) + \cot\left( \frac{\alpha_1 - \theta_j}{2} \right)
               \right)
         \quad = \quad 0.
\end{equation}
Fix any $\theta\geq 0$ and define the monotonically decreasing and
continuous function $f:(0,\theta_{1})\mapsto\mathbb{R}$:
\[
f(x) \quad = \quad
a_0^2 \cot\frac{x}{2} 
             + \sum_j a_j^2 \left(\cot \left(\frac{x + \theta_j}{2}
               \right) + \cot\left( \frac{x - \theta_j}{2} \right)
               \right)
             - \tan^2\theta \tan \frac{x}{2}.
\]
\suppress{
We are looking for the only value $\alpha_1^\theta\in(0,\theta_{1})$
such that $f(\alpha_1^\theta)=0$.
}
We know that $\lim_{x\to 0+} f(x)=+\infty$,
$\lim_{x\to\theta_{1}-} f(x)=-\infty$, and $f(\alpha_1)\leq 0$. 
Therefore there is a unique $\alpha_1^\theta \in 
(0,\alpha_1]$ such that~$f(\alpha_1^\theta) = 0$.

{From} the monotonicity of $\cot$, 
for~$x \in (0,\alpha_1]$, we have
$$\sum_j a_j^2 \left(\cot \left(\frac{x + \theta_j}{2}
               \right) + \cot\left( \frac{x - \theta_j}{2} \right)
               \right)
\quad \geq \quad 
\sum_j a_j^2 \left(\cot \left(\frac{\alpha + \theta_j}{2}
               \right) + \cot\left( \frac{\alpha - \theta_j}{2} \right)
               \right)
.$$
Using this inequality together with
Eq.~\eqref{alpha}
and the monotonicity of $\tan$, 
we bound the function $f$ from above as follows:
\begin{eqnarray*}
f(x)&\geq &a_0^2 \cot\frac{x}{2}-a_0^2 \cot\frac{\alpha}{2}- \tan^2\theta \tan \frac{x}{2}\\
&\geq &a_0^2 \cot\frac{x}{2}-a_0^2 \cot\frac{\alpha}{2}- \tan^2\theta \tan \frac{\alpha}{2}\\
&\geq &a_0^2 \cot\frac{x}{2}-1.01\times a_0^2 \cot\frac{\alpha}{2},
\end{eqnarray*}
where the last inequality comes from the hypothesis
$0\leq\tan\theta\leq a_0\cot(\alpha_1/2)/10$.
Since $f(\alpha_1^\theta)=0$, we get that 
$\cot(\alpha_1^\theta/2)\leq 1.01\times \cot(\alpha_1/2)$.

We now  prove that
$f(0.78\times \alpha_1)\leq 0$, which concludes the proof.
In the rest of the proof, we restrict the variable~$x$ to
the interval $[0.78\times \alpha_1,\alpha_1]$.
Let $\beta$ be the solution of $\tan(\beta/2)= \tan(\alpha_1/2)/\sqrt{1.01}$
in~$[0,\pi/2)$.
Then $f(\beta)\geq 0$, and therefore $\alpha_1^\theta\geq\beta$.

Since $\alpha_1\geq 0$, we have $\tan(\alpha_1/2)\geq \alpha_1/2$ and
$$
\tan\frac{\beta}{2}
\quad = \quad \frac{\tan(\frac{\alpha_1}{2})}{\sqrt{1.01}}
\quad \geq \quad \frac{\alpha_1}{2\sqrt{1.01}}.
$$
Moreover since $0\leq\beta\leq\alpha_1\leq\pi/4$, we have
$\tan (\beta/2)\leq 2\beta/\pi$, and therefore
$$\beta \quad \geq \quad \frac{\pi}{4\sqrt{1.01}}\times\alpha
\quad \geq \quad 0.78\times \alpha_1.$$
\end{proof}

\section{Proof of Lemma~\ref{lem-approximation2}}\label{applem-approximation2}
\begin{proof}
Since the two vectors $\ket{w_{\pm\alpha_1,\theta}}$ 
(respectively, $\ket{w_{\pm\alpha_1}}$)
have the same norm and are orthogonal,
it suffices to upper bound the following squared norm:
\begin{equation}\label{long-expression}
\norm{w_{\alpha_1,\theta} - w_{-\alpha_1,\theta}}^2 \quad = \quad 4 \cos^2 \theta \; \left[a_0^2
  \cot^2 \alpha_{1}^\theta/2 + 2 \sum_j a_j^2 \left(\frac{\sin
  \alpha_{1}^\theta}{\cos \alpha_{1}^\theta - \cos \theta_j} \right)^2 + \tan^2 \theta
  \; \tan^2 \frac{\alpha_{1}^\theta}{2} \right]
\end{equation}
By Lemma \ref{lem-approximation1}, $$a_0^2 \cot^2 \alpha_{1}^\theta/2
\quad \leq \quad 1.01 \; a_0^2 \cot^2 \alpha_1/2$$ and
$$\tan^2 \theta \; \tan^2 \frac{\alpha_{1}^\theta}{2} \quad = \quad
0.01 \; a_0^2 \cot^2 \frac{\alpha_1}{2} \; \tan^2 \frac{\alpha_{1}^\theta}{2}
\quad \leq \quad 0.01 \; a_0^2$$
by the choice of $\tan \theta = a_0 \cot(\alpha_1/2)/10$ and the
monotonicity of $\cot^2 = 1/\tan^2$ on $(0,\pi/2)$.  Since
\begin{equation}
\sum_j a_j^2 \left(\frac{\sin
  \alpha_{1}^\theta}{\cos \alpha_{1}^\theta - \cos \theta_j} \right)^2
  \quad \leq \quad \sum_j a_j^2 \left(
  \frac{\sin \alpha_1}{\cos \alpha_1 - \cos \theta_j} \right)^2,
\end{equation}
the quantity inside of the square
brackets of (\ref{long-expression}) is at most:
\begin{equation*}
2 \left( a_0^2 \cot^2 \alpha_1/2 + 2 \sum_j a_j^2 \left(
\frac{\sin \alpha_1}{\cos \alpha_1 - \cos \theta_j} \right)^2 + a_0^2 \right).
\end{equation*}
Hence,
\begin{equation*}
\norm{w_{\alpha_1,\theta} - w_{-\alpha_1,\theta}}^2 \quad \leq \quad
(8 \cos^2 \theta) \times  \norm{w_{\alpha_1} - w_{-\alpha_1}}^2.
\end{equation*}
\end{proof}

\end{document}